\pdfoutput=1
\newif\ifdouble
\doubletrue
\ifdouble
\documentclass[journal]{IEEEtran}
\else
\documentclass[12pt,draftclsnofoot,onecolumn]{IEEEtran}
\fi

\usepackage[cmex10]{amsmath}
\usepackage{amsfonts}%
\usepackage{amssymb}%
\usepackage{amsthm}
\usepackage{hyperref}
\usepackage{cleveref}
\usepackage[pdftex]{graphicx}
\usepackage{epstopdf}
\usepackage{psfrag}
\usepackage{tipa}
\usepackage[boxed]{algorithm2e}
\usepackage{array}
\usepackage{mdwmath}
\usepackage{mdwtab}
\usepackage{eqparbox}
\usepackage{subfigure}
\usepackage{fullpage}
\usepackage{cite}

\newtheorem{theorem}{Theorem}
\newtheorem{lemma}{Lemma}

\newtheorem{corollary}{Corollary}

\newtheorem{claim}{Claim}

\newcommand{\off}[1]{}

\newcommand{\E}{\mathrm{E}}

\newcommand{\bw}{\mathbf{w}}

\newcommand{\by}{\mathbf{y}}
\newcommand{\bH}{\mathbf{H}}
\newcommand{\bh}{\mathbf{h}}

\newcommand{\bM}{\mathbf{M}}
\newcommand{\bR}{\mathbf{R}}
\newcommand{\bn}{\mathbf{n}}
\newcommand{\bz}{\mathbf{z}}
\newcommand{\bG}{\mathbf{G}}
\newcommand{\bg}{\mathbf{g}}

\newcommand{\bx}{\mathbf{x}}

\newcommand{\bgamma}{\pmb{\gamma}}

\newcommand{\bxi}{\pmb{\xi}}

\newcommand{\balpha}{\pmb{\alpha}}
\newcommand{\bbeta}{\pmb{\beta}}

\newcommand{\prnt}[1]{\left(#1\right)}
\newcommand{\brkt}[1]{\left[#1\right]}
\newcommand{\brcs}[1]{\left\{#1\right\}}

\newcommand{\SEFI}[1]{{\bfseries *** Sefi says: #1 ***}}

\begin{document}
\title{On Secrecy Rates and Outage in Multi-User  Multi-Eavesdroppers MISO Systems}
\author{
Joseph Kampeas, Asaf Cohen and Omer Gurewitz
\\
Ben-Gurion University of the Negev, Beer-Sheva, 84105, Israel
\\ {\tt Email: \{kampeas,coasaf,gurewitz\}@bgu.ac.il} }
\date{}
\maketitle
\begin{abstract}
In this paper, we study the secrecy rate and outage probability in Multiple-Input-Single-Output (MISO) Gaussian wiretap channels at the limit of a large number of legitimate users and eavesdroppers. In particular, we analyze the asymptotic achievable secrecy rates and outage,  when \emph{only statistical knowledge} on the wiretap channels is available to the transmitter.

The analysis provides exact expressions for the reduction in the secrecy rate as the \emph{number of eavesdroppers grows}, compared to the boost in the secrecy rate as \emph{the number of legitimate users grows}.
\end{abstract}
\maketitle

\section{Introduction}
The explosive expansion of wireless communication and wireless based services is leading to a growing necessity to provide privacy in such systems. Due to the broadcast nature of the transmission, wireless networks are inherently susceptible to eavesdropping. One of the most promising techniques to overcome this drawback is utilizing Physical-Layer security. Physical-layer security leverages the random nature of communication channels to enable encoding techniques such that eavesdroppers with inferior channel quality are unable to extract any information about the transmitted information from their received signal \cite{wyner1975wire, csiszar1978broadcast, leung1978gaussian}.

Many recent studies have explored the potential gains in exploiting multiple antenna technology for attaining secrecy in various setups. For example, in the case of a single eavesdropper, when the transmitter has a full Channel State Information (CSI) on the wiretap channel, it can ensure inferior wiretap channel by nulling the reception on the eavesdropper's end, thus, achieve higher secrecy rates \cite{shafiee2007achievable,shafiee2009towards,geraci2011secrecy}. When the user and eavesdropper are also equipped with multiple antennas, the optimal strategy is to utilize linear precoding in order to focus energy only in few directions, thus achieving the optimal secrecy rate \cite{oggier2011secrecy}.
In case that \emph{only statistical information} on the wiretap channel is available, the optimal scheme is beamforming in the user's direction~\cite{shafiee2007achievable,li2011ergodic}. However, in this case, a \emph{secrecy outage}, the event that at the eavesdropper is able to extract all or part of the message, is unavoidable. To mitigate risk, one should consider transmitting Artificial Noise (AN) to further degrade the wiretap channel \cite{khisti2010secure,khisti2010secure2,fakoorian2013optimality,6678299,7005544}. 

The secrecy capacity \emph{at the limit of large number of antennas} was considered in \cite{khisti2010secure,khisti2010secure2}. Particularly, \cite{khisti2010secure,khisti2010secure2} studied the asymptotic (in the number of cooperating antennas) secrecy capacity, for a single receiver. 
\cite{krikidis2013secrecy} used Extreme Value Theory (EVT) to study the scaling law of the secrecy sum-rate under a random beamforming scheme, when the users and eavesdroppers \emph{are paired}. That is, each user was susceptible to eavesdropping only by its paired (single) eavesdropper. 
\cite{geraci2014secrecy} considered the asymptotic secrecy rate, where both the number of users and number of antennas grow to infinity, while all users are potentially malicious and few external eavesdroppers are wiretapping to the transmissions.

In the presence of a single user and multiple eavesdroppers, where the transmitter has no CSI on the wiretap channels, a secrecy outage will definitely occur as the number of eavesdroppers goes to infinite \cite{wang2007secrecy}. On the other hand, when there are many legitimate users, and the transmitter can select users opportunistically, the secrecy outage probability is open in general. In particular, the asymptotically exact expression to the number of users required in order to attain sufficiently small secrecy outage probability, is yet to be solved. 
 
 This study analyzes this subtle relation between the number of users, eavesdroppers and the resulting secrecy outage probability.
Specifically, we consider the secrecy rate and outage probability for the Gaussian MISO wiretap channel model, where a transmitter is serving $K$ legitimate users in the presence of $M$ eavesdroppers, and analyze the secrecy outage probability as a function of $K$ and $M$, and more importantly, the relation between these two numbers. 

We assume that CSI is available from all legitimate users, yet \emph{only channel statistics} are available on the eavesdroppers. As previously mentioned, when the  transmitter has only statistical information on the wiretap channel, transmitting in the direction of the attending user is optimal when AN is not allowed. Moreover, in large scale systems, using AN may interfere with other cells, and probably would not be a method of choice even at the price of reduced secrecy rate.
Beamforming to the attending user, on the other hand, is the de-facto transmission scheme in many MISO systems today. Therefore, we adopt the scheme in which at each transmission opportunity the transmitter beamforms in the direction of a user with favorable channel. We analyze the asymptotics of the secrecy rate and secrecy outage under the aforementioned scheme. In particular, our contributions are as follows:
(i)	We first analyze the secrecy rate distribution when transmitting to the strongest user while many eavesdroppers are wiretapping. These results are utilized to attain the secrecy outage probability in the absence of the wiretap channels' CSI. 
(ii) We provide both upper and lower bounds on the limiting secrecy rate distribution. The bounds are tractable and give insight on the scaling law and the effect of the system parameters on the secrecy capacity. 
 We show via simulations that our bounds are tight.
(iii) We quantify the reduction in the secrecy rate as the \emph{number of eavesdroppers grows}, compared to the boost in the secrecy rate as \emph{the number of legitimate users grows}. We show that in order to
attain asymptotically small secrecy outage probability with $t$ transmit antennas, $\Omega\prnt{n \prnt{\log n}^{t-1}}$ users are required in order to compensate for $n$ eavesdroppers in the system. 


\section{System Model}
Throughout this paper, we use bold lower case letters to denote random variables and random vectors, unless stated otherwise. $V^\dagger$ denotes the Hermitian transpose of matrix $V$. Further, $\vert \cdot \vert$, $\langle \cdot, \cdot \rangle$ and $\Vert \cdot \Vert$ denote the absolute value of a scalar, the inner product and the Euclidean norm of vectors, respectively.

  Consider a MISO downlink channel with one transmitter  with $t$ transmit antennas, $K$ legitimate users with a single antenna and $M$ uncooperative eavesdroppers, again, with one antenna each.
The transmitter adopts the scheme in which at each transmission opportunity the transmitter beamforms in the direction of the selected user without AN.
 We assume a block fading channel  where the transmitter can query for fine channel reports from the users before each transmission, while having \emph{only statistical knowledge on the wiretap channels}. 

Let $\by_{i}$ and $\bz_{j}$ denote the received signals at user $i$ and at eavesdropper $j$, respectively. Then, the received signals can be described as $\by_{i} = \bh_{i}\bx + \bn_{b(i)}$ and $\bz_{j} = \bg_{j}\bx + \bn_{e(j)},$ where $\bh_{i} \in \mathbb{C}^{t \times 1}$ and $\bg_{j} \in \mathbb{C}^{t \times 1}$ are the channel vectors between the transmitter and user $i$, and between the transmitter and  eavesdropper $j$, respectively. $\bh_{i}$ and $\bg_{j}$ are random complex Gaussian channel vectors, where the entries have zero mean and  unit variance in the real and imaginary parts.  
 $\bx \in \mathbb{C}^t$ is the transmitted vector, with a power constraint $\E\brkt{\bx^\dagger \bx} \leq P$, while  $\bn_{b(i)}, \bn_{e(j)} \in \mathbb{C}$ are unit variance Gaussian noises seen at user $i$ and eavesdropper $j$, respectively.

The secrecy capacity for the Gaussian MIMO wiretap channel, where the main and wiretap channels, $\bH$ and $\bG$, respectively, are known at the transmitter,  was given in  \cite{oggier2011secrecy,khisti2010secure2}

{\small{
\begin{equation}
C_s = \max_{\Sigma_{\bx}} \log\det\prnt{I + \bH \Sigma_{\bx} \bH^\dagger } - \log\det\prnt{I + \bG \Sigma_{\bx} \bG^\dagger }
\label{eq: MIMO secrecy capacity}
\end{equation}
}}\normalsize

with $\mathrm{tr}\prnt{\Sigma_{\bx}} \leq P$. For the special case of Gaussian MISO wiretap channel, (\ref{eq: MIMO secrecy capacity}) reduces to
{\small{
\begin{equation*}
C_s = \max_{\Sigma_{\bx}} \log\det\prnt{I + \bh \Sigma_{\bx} \bh^\dagger } - \log\det\prnt{I + \bg \Sigma_{\bx} \bg^\dagger }.
\end{equation*}
}\normalsize}
In both Gaussian MIMO and MISO, the optimal $\Sigma_{\bx}$ is \emph{low rank}, which means that to achieve the secrecy capacity, the optimal strategy is \emph{transmitting in few directions}. Specifically, for the  Gaussian MISO wiretap channel, the capacity achieving strategy is beamforming to a single direction, hence, letting $\bw$ denote a beam vector, then $\Sigma_{\bx} = \bw \bw^\dagger$, \cite{shafiee2007achievable}. Moreover, when the wiretap channel is unknown at the transmitter, it is optimal to beamform in the direction of the main channel, i.e., $\bw = \hat{\bh} = \bh / \Vert \bh \Vert$, \cite{shafiee2007achievable,li2011ergodic}.

Accordingly, when beamforming in the direction of the user  while the eavesdropper is wiretapping, assuming only the main channel is known to transmitter,  the secrecy capacity is \cite{wang2007secrecy,barros2006secrecy}:
\begin{equation}
\bR_{s}(\bh,\bg) = \log \prnt{\frac{1 +  P \Vert \bh \Vert^2}{ 1 +  P \vert \langle \hat{\bh}, \bg\rangle \vert^2}}.
\label{eq: sec capacity general}
\end{equation}
 Recall that in the block fading environment, $\bh$ and $\bg$ are random variables and are drawn from  the Gaussian distribution independently after each block (slot). Hence, the distribution of the ratio in (\ref{eq: sec capacity general}) and its support are critical to obtain important performance metrics. In particular, the \emph{ergodic secrecy rate}, i.e., the secrecy rate when considering coding over a large number of time-slots, can be obtained by computing an expectation with respect to the fading of both $\bg$ and $\bh$. Similarly, a certain target secrecy rate $R_s$ is achievable if the instantaneous ratio in (\ref{eq: sec capacity general}) is greater than the matching value.
On the other hand, a \emph{secrecy outage} occurs if $R_s$ is greater than the instantly achievable secrecy rate $\bR_{s}(\bh,\bg)$, and thus, the message cannot be delivered securely~\cite{barros2006secrecy}. The probability of such event is $\Pr \prnt{\bR_{s}(\bh,\bg) < R_s}$.

For clarity, let us point out a few statistical properties of the ratio in (\ref{eq: sec capacity general}). In the denominator, the squared inner product  $\vert \langle \hat{\bh}, \bg\rangle \vert^2$ follows the Chi-squared distribution with $2$ degrees of freedom, $\chi^2(2)$ (which is equivalent to the Exponential distribution with rate parameter 1/2), since $\hat{\bh}$ is normalized,  rotating both $\hat{\bh}$ and $\bg$ such that $\hat{\bh}$ aligns with the unit vector does not change the inner product. Thus, the inner product result in a complex Gaussian random variable \cite{sharif2005capacity,jagannathan2006efficient}.
Similarly, in the numerator, the squared norm $\Vert \bh  \Vert^2$ follows the Chi-squared distribution with $2t$ degrees of freedom, $\chi^2(2t)$, since it is a sum of $t$ squared complex Gaussian random variables. Thus, for any user $i$ and eavesdropper $j$, the secrecy SNR  when beamforming to user $i$ is equivalent to the ratio of $1+\chi^2(2t)$ and $1+\chi^2(2)$ random variables.

\subsection{Main Tool}
To assess the ratio in the presence of large number of users and eavesdroppers, let us recall that for sufficiently large $n$, the maximum of a sequence of $n$ i.i.d.\ $\chi^2 (v)$ variables,  $\bM_n = \max \prnt{\bxi_1,...,\bxi_n}$ follows the Gumbel distribution \cite[pp. 156]{embrechts2011modelling}. Specifically,
$\lim_{n \to \infty}\Pr \prnt{\bM_n \leq a_n \xi + b_n} = \exp\brcs{-e^{-\xi}}$,
where $a_n$ and $b_n$ are normalizing constants. In this case, 
\begin{align}
&a_{n} =   2, &\label{eq: a_n}\\
&b_{n} = 2\prnt{\log n + \prnt{\frac{v}{2}-1}\log \log n - \log \Gamma\brkt{\frac{v}{2}}}&\label{eq: b_n}
\end{align}
and $\Gamma[\cdot]$ is the Gamma function. 

In this paper, we study the asymptotic (in the number of users and eavesdroppers) distribution of the ratio in (\ref{eq: sec capacity general}), and thus derive the secrecy outage probability, \emph{when the transmitter schedules a user with favorable CSI} and beamforms in its direction.

\section{Asymptotic Secrecy Outage }\label{sec. limit distribution}
%
In this section, we analyze the secrecy outage limiting distribution. That is, for a given target secrecy rate $R_s$, we analyze the probability that \emph{at least one eavesdropper among $M$ eavesdroppers} will attain information from the transmission. 
Obviously, when transmitting to a \emph{single user}, and when only statistical knowledge is available on the wiretap channels, beamforming to the user whose channel gain is the greatest among $K$ users is optimal.

Accordingly, let  $i^* = \arg\max_i \Vert \bh_{i}\Vert^2$ be the index of the channel with the largest gain, and let $j^* = \arg\max_j \vert \langle \hat{\bh}_{i^*}, \bg_{j}\rangle \vert^2$ be the index of the wiretap channel whose projection in the direction $\hat{\bh}_{i^*}$ is the largest. 
Note that when the transmitter beamforms to user $i^*$ in a multiple eavesdroppers environment, it should tailor a code with secrecy rate $R_s$ to protect the message even from the strongest eavesdropper with respect to $i^*$, which is $j^*$. 
Of course, with only statistical information on the eavesdroppers, $j^*$ is unknown to the transmitter.
 Accordingly, the probability of a secrecy outage when transmitting to user $i^*$ at secrecy rate $R_s$ is \cite{barros2006secrecy}:
\begin{align}
&\Pr \prnt{\log_2\prnt{\frac{1+ P\Vert \bh_{i^*}\Vert^2}{1+ P \vert \langle \hat{\bh}_{i^*}, \bg_{j^*}\rangle \vert^2}} \leq R_s} & \nonumber\\
&=  \Pr \prnt{\frac{1+ P\Vert \bh_{i^*}\Vert^2}{1+ P \vert \langle \hat{\bh}_{i^*}, \bg_{j^*}\rangle \vert^2} \leq 2^{R_s}}&
\label{eq: sec capacity MRT}
\end{align}
To ease notation, we denote $\alpha = 2^{R_s}$.

In the following, we analyze the distribution in (\ref{eq: sec capacity MRT}) when the number of users and eavesdroppers is large. In particular, we consider the secrecy rate distribution when the transmitter  beamforms to user $i^*$, while all eavesdroppers are striving to intercept the transmission separately (without cooperation).

When the transmitter is beamforming to a user whose channel gain is the greatest, then the squared norm $\Vert\bh_{i^*}\Vert^2$ in the numerator of (\ref{eq: sec capacity MRT}) scales with the number of users like $O(\log K)$ \cite{embrechts2011modelling}. Nevertheless, the greatest channel projection in the direction of the attending user, in the denominator of (\ref{eq: sec capacity MRT}), also scales with the number of eavesdroppers in the order of  $O(\log M)$. Moreover, asymptotically, both the greatest gain and greatest channel projection follow the Gumbel distribution (with different normalizing constants). Thus, in order to determine the secrecy rate behavior, as $K$ and $M$ grow, one needs to address the ratio of Gumbel random variables. However, the ratio distribution of Gumbel random variables is not known to have a closed-form \cite{nadarajah2006ratios}. Thus, we first express it as an infinite sum of Gamma functions, then provide tight bounds on the obtained distribution, from which we can infer the outage probability. Accordingly, we have the following.

\begin{theorem}\label{theorem: mu-me exact outage}
For large enough $K$ and $M$, the distribution of the secrecy rate in (\ref{eq: sec capacity MRT})  is the following.
\begin{multline*}
\Pr\prnt{\frac{1+ P\Vert \bh_{i^*}\Vert^2}{1+ P \vert \langle \hat{\bh}_{i^*}, \bg_{j^*}\rangle \vert^2} \leq \alpha}\\
 =  \sum_{k=0}^\infty \frac{(-1)^k e^{-(k+1)\frac{1+b_K - \alpha(1+b_M)}{\alpha a_M}}}{(k+1)!} \Gamma\left[1+\frac{(k+1)a_K}{\alpha a_M}\right]
\end{multline*}
where $a_K$, $a_M$ and $b_K$, $b_M$ are normalizing constants given in (\ref{eq: a_n}) and (\ref{eq: b_n}), respectively.
\end{theorem}
Note that $b_K$ and $b_M$ grow at different rate. Specifically, although $b_K$ and $b_M$ are both normalizing constant of the $\chi^2$ distribution, $b_K$ has value of $v=2t$ in (\ref{eq: b_n})), while $b_M$ has value of $v=2$ in (\ref{eq: b_n}).  
The proof is given in the Appendix.

To evaluate the result in Theorem~\ref{theorem: mu-me exact outage}, one needs to evaluate the infinite sum, which is intricate. Thus, the following upper and lower bounds are useful.
\subsection{Bounds on the Secrecy Rate Distribution}
In the following, 
 we suggest an approach that models EVT according to its  tail distribution, which enables us to provide tight bounds to the distribution in~(\ref{eq: sec capacity MRT}). This approach has very clear and intuitive \emph{communication interpretation}. Specifically, for an upper bound, we put a threshold on eavesdropper $j^*$'s wiretap channel projection, and analyze the result under the assumption that its projection exceeded. For a lower bound, we put a threshold on user $i^*$'s channel gain, and analyze the result under the assumption that its gain has exceeded it.

When only a single user (eavesdropper), among many, exceeds a threshold on average, then the above-threshold tail distribution corresponds to the tail of extreme value distribution \cite[Ch. 4.2]{coles2001introduction}.
   Moreover, the tail limiting distribution has a mean value that is higher than the mean value of the extreme value distribution, since the tail limiting distribution takes into account only events in which user $i^*$ (eavesdropper $j^*$) is sufficiently strong, namely, above threshold. Thus, replacing the extreme value distribution of user $i^*$ (eavesdropper $j^*$) with its corresponding tail distribution will increase the numerator  (denominator) in (\ref{eq: sec capacity MRT}) on average. Thus, the resulting secrecy rate is  higher (lower), hence, corresponds to a lower (upper) bound on the ratio CDF.

Let $u_m$ denote a threshold on the wiretap channel projection in the direction $\hat{\bh}_{i^*}$, such that a single (the strongest) eavesdropper exceeds it on average. Note that such a threshold can be obtained by inversing the complement CDF of the Exponential distribution. Further, note that this inverse is exactly (\ref{eq: b_n}) with $v=2$ degrees of freedom.  Accordingly, we have the following lower bound.
\begin{lemma}\label{lemma: some eve above thr}
For sufficiently large $K$ and $M$, the CDF of the secrecy rate  in (\ref{eq: sec capacity MRT}) satisfies the following upper bound.
\begin{align*}
&\Pr \prnt{\frac{1+ P\Vert \bh_{i^*}\Vert^2}{1+ P\vert \langle \hat{\bh}_{i^*}, \bg_{j^*}\rangle \vert^2} \leq \alpha} &\\
& \leq \frac{a_K}{\alpha a_M} e^{-\frac{1+ b_K- \alpha(1+ u_m) }{\alpha a_M}} \Gamma \brkt{\frac{a_K}{\alpha a_M},0, e^{\frac{1+b_K - \alpha(1+ u_m ) }{a_K}}},&
\end{align*}
where $\Gamma\brkt{s,0,z} = \int_0^z \tau^s e^{-\tau} \mathrm{d}\tau$ is the lower incomplete Gamma function.
\end{lemma}
The proof is given in the Appendix. 
The following corollary helps gaining insights from Lemma~\ref{lemma: some eve above thr}.
\begin{corollary}\label{coro: meaningful upper bound}
For $\alpha \geq 1$, the outage probability satisfies the following bound.
\begin{multline*}
\Pr \prnt{\frac{1+ P\Vert \bh_{i^*}\Vert^2}{1+ P\vert \langle \hat{\bh}_{i^*}, \bg_{j^*}\rangle \vert^2} \leq \alpha} \\
< \brkt{\Lambda(\alpha)\prnt{1-\exp\brcs{-\Lambda(\alpha)^{-1} 2^{\alpha -1}}}}^{1/\alpha},
\end{multline*}
where $\Lambda(\alpha) =  \prnt{\sqrt{e}M}^\alpha\frac{ \Gamma(t)}{\sqrt{e}K \prnt{\log K}^{t-1}}$.
\end{corollary}
Note that the value of $\Lambda(\alpha)$ determines the outage probability. In particular, at the limit $\Lambda(\alpha) \to \infty$, the resulting outage probability is $1$ (i.e., when $M \to \infty$ and $K$ is fixed). Similarly, when $\Lambda(\alpha) \to 0$, the resulting outage probability $0$. Moreover, we point out that $\Lambda(\alpha)$ is decreasing with the number of users as $K \prnt{\log K}^{t-1}$, while increasing with the number of eavesdroppers as $M^\alpha$. 
Thus, roughly speaking, as long as the number of eavesdroppers $M = o\prnt{K \prnt{\log K}^{t-1}}^{1/\alpha}$, we obtain $\Lambda(\alpha) = o(1)$, hence, secrecy outage in the order of $o(1)$. 

To prove Corollary~\ref{coro: meaningful upper bound}, the following Claim is useful. 
\begin{claim}[\cite{gautschi1998incomplete,laforgia2013some,mortici2010new}]\label{claim: gamma function}
The incomplete Gamma function satisfies the following bounds.
\begin{description}
	\item[$(i)$] $ \Gamma\brkt{s} \prnt{1-e^{-z}}^{s} < \Gamma\brkt{s,0,z}<  \Gamma\brkt{s} \prnt{1-e^{-z \Gamma\brkt{1 + s}^{-1/s} }}^{s}, \forall\quad 0 < s < 1$. 
	This inequality takes the opposite direction for values of $s>1$.
	\item[$(ii)$] $2^{s-1}\leq \Gamma\brkt{1+s} \leq 1, \forall\quad 0 < s < 1$.
	\item[$(iii)$] $\Gamma\brkt{s}\Gamma\brkt{1/s} \geq 1, \forall s>0$.
	%
\end{description}
\end{claim}
\begin{proof}{(Corollary~\ref{coro: meaningful upper bound}).}
Applying the normalizing constants in (\ref{eq: a_n})-(\ref{eq: b_n}) to Lemma~\ref{lemma: some eve above thr} result in
\begin{align*}
&\Pr \prnt{\frac{1+ P\Vert \bh_{i^*}\Vert^2}{1+ P\vert \langle \hat{\bh}_{i^*}, \bg_{j^*}\rangle \vert^2} \leq \alpha} &\\
& \leq \frac{1}{\alpha}\prnt{\frac{\prnt{\sqrt{e}M}^\alpha \Gamma[t] }{\sqrt{e}K \prnt{\log K}^{t-1}}}^{\frac{1}{\alpha}} \Gamma \brkt{\frac{1}{\alpha},0, \frac{\sqrt{e}K \prnt{\log K}^{t-1}}{\prnt{\sqrt{e}M}^\alpha \Gamma[t] } }&
\end{align*}
To ease notation, let us denote 
$\Lambda(\alpha) =  \frac{\prnt{\sqrt{e}M}^\alpha \Gamma(t)}{\sqrt{e}K \prnt{\log K}^{t-1}}$.
 Thus, we rewrite Lemma~\ref{lemma: some eve above thr} as
\begin{align*}
&\Pr \prnt{\frac{1+ P\Vert \bh_{i^*}\Vert^2}{1+ P\vert \langle \hat{\bh}_{i^*}, \bg_{j^*}\rangle \vert^2} \leq \alpha} &\\
 &\leq \frac{1}{\alpha}\Lambda(\alpha)^{1/\alpha} \Gamma \brkt{\frac{1}{\alpha},0, \Lambda(\alpha)^{-1}}&\\
 &\stackrel{(a)}{<} \frac{1}{\alpha}\Lambda(\alpha)^{1/\alpha} \Gamma\brkt{\frac{1}{\alpha}}\prnt{1-e^{-\frac{\Gamma\brkt{1 + \frac{1}{\alpha}}^{-\alpha}}{\Lambda(\alpha)}}}^{1/ \alpha}&\\
& \stackrel{(b)}{\leq} \brkt{\Lambda(\alpha)\prnt{1-e^{-\frac{2^{\alpha -1}}{\Lambda(\alpha)}}}}^{1/ \alpha}&
\end{align*}
Remember that only $\alpha \geq 1$ implies secrecy rate greater than zero. Thus, $(a)$ follows from Claim~\ref{claim: gamma function}$(i)$, and $(b)$ follows from Claim~\ref{claim: gamma function}$(ii)$ and from the Gamma function recurrence property, $\Gamma\brkt{\frac{1}{\alpha}} = \alpha\Gamma\brkt{1+ \frac{1}{\alpha}}$.

\end{proof}
For the lower bound, we utilize a similar approach, however, this time, we refer to user  $i^*$ as if its channel gain has exceeded a high threshold. 
Thus, since only sufficiently strong user $i^*$, whose gain is above threshold, is taken into account, then the numerator in (\ref{eq: sec capacity MRT}) is larger on average, thus, resulting  in a higher rate, which corresponds to a lower bound on the ratio CDF.

Let $u_k$ denote a threshold on the user's channel gain, such that a single strongest user exceeds it on average. Note that such a threshold can be obtained from the inverse incomplete Gamma function, which asymptotically, is exactly (\ref{eq: b_n}) with $v=2t$.  
Accordingly, we have the following upper bound.
\begin{lemma}\label{lemma: some bob above thr}
For sufficiently large $K$ and $M$, the CDF of the secrecy rate  in (\ref{eq: sec capacity MRT}) satisfies the following lower bound.
\begin{multline*}
\Pr \prnt{\frac{1+ P\Vert \bh_{i^*}\Vert^2}{1+ P\vert \langle \hat{\bh}_{i^*}, \bg_{j^*}\rangle \vert^2} \leq \alpha} \geq 1 - \frac{\alpha a_M}{a_K}\\
 \quad \cdot  e^{-\frac{\alpha(1+ b_M) - (1+ u_k) }{a_K}} \Gamma \brkt{\frac{\alpha a_M}{a_K},0, e^{-\frac{1+u_k -\alpha(1+ b_M)}{\alpha a_M}}}.
\end{multline*}
\end{lemma}
The proof is given in the Appendix.

Again, to gain intuition, we have the following.
\begin{corollary}\label{coro: meaningful lower bound}
For $\alpha \geq 1$, the outage probability satisfies the following bound.
\begin{multline*}
\Pr \prnt{\frac{1+ P\Vert \bh_{i^*}\Vert^2}{1+ P\vert \langle \hat{\bh}_{i^*}, \bg_{j^*}\rangle \vert^2} \leq \alpha} \\
>1- \Gamma\brkt{1+\alpha}\Lambda(\alpha)^{-1}\prnt{1-e^{-\Lambda(\alpha)^{1/\alpha}}  }^{\alpha},
\end{multline*}
where $\Lambda(\alpha) =  \prnt{\sqrt{e}M}^\alpha\frac{\Gamma(t)}{\sqrt{e}K \prnt{\log K}^{t-1}}$.
\end{corollary}
\begin{proof}
Similar to Corollary~\ref{coro: meaningful upper bound}, we apply the normalizing constants in (\ref{eq: a_n})-(\ref{eq: b_n}) to Lemma~\ref{lemma: some bob above thr}, then, set $\Lambda(\alpha) =  \frac{\prnt{\sqrt{e}M}^\alpha \Gamma(t)}{\sqrt{e}K \prnt{\log K}^{t-1}}$.
Thus, we have
\begin{align*}
&\Pr \prnt{\frac{1+ P\Vert \bh_{i^*}\Vert^2}{1+ P\vert \langle \hat{\bh}_{i^*}, \bg_{j^*}\rangle \vert^2} \leq \alpha}&\\ 
 &\geq 1 - \alpha \Lambda(\alpha)^{-1} \Gamma \brkt{\alpha,0, \Lambda(\alpha)^{1/\alpha}}&\\
 &\stackrel{(a)}{>} 1 - \alpha \Lambda(\alpha)^{-1} \Gamma\brkt{\alpha}\prnt{1-e^{-\Lambda(\alpha)^{1/\alpha}}}^{\alpha}&\\
& \stackrel{(b)}{=} 1- \Gamma\brkt{1+\alpha}\Lambda(\alpha)^{-1}\prnt{1-e^{-\Lambda(\alpha)^{1/\alpha}}  }^{\alpha}&
\end{align*}
where $(a)$ follows form Claim~\ref{claim: gamma function}$(i)$ and $(b)$ follows from the Gamma function recurrence property.
\end{proof}


\section{Simulation Results}
In this section, we present simulate results for the suggested scheduling scheme and compare them to the analysis.

Figure~\ref{fig: ratio_approx_distribution} depicts the distribution of the ratio in (\ref{eq: sec capacity MRT}) for three cases and compare it to the analytical results herein. In particular, we simulate the secrecy rate in (\ref{eq: sec capacity general}) for three cases: (i) When beamforming to strongest user $i^*$, while the strongest, above-threshold, eavesdropper is wiretapping, (ii) when beamforming to strongest user $i^*$, while strongest eavesdropper $j^*$ is wiretapping (without threshold constraint). (iii) when beamforming to strongest, above-threshold, user, while the strongest eavesdropper $j^*$ is wiretapping. The dark gray, gray and light gray bars represents these results, respectively. Then we evaluate the bounds given in Lemma~\ref{lemma: some eve above thr} and Lemma~\ref{lemma: some bob above thr} and compare then to the sum of the first $100$ terms in Theorem~\ref{theorem: mu-me exact outage}, for $t=2,4,8$ and $M=K=30$. It is clear the bounds are tight and provide excellent approximation to (\ref{eq: sec capacity MRT}).
\begin{figure}
	\centering
		\includegraphics[scale=0.85]{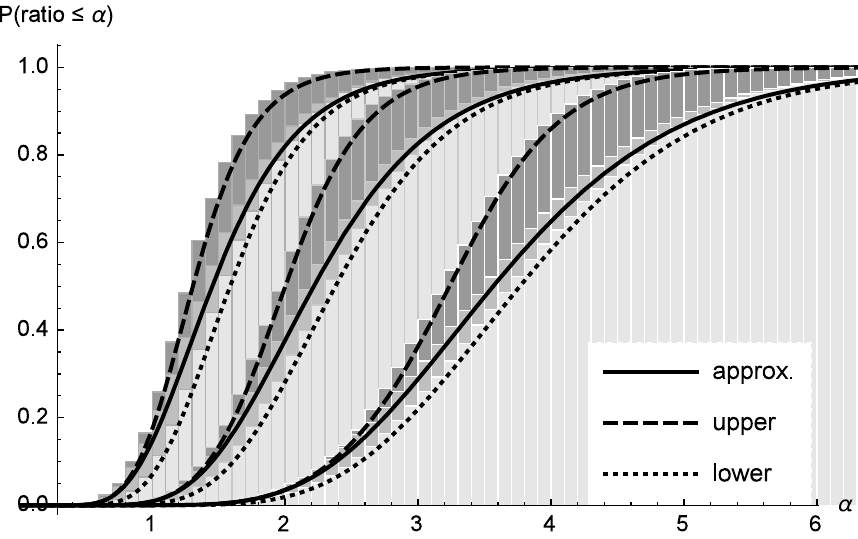}
	\caption{Simulation and analysis of the  ratio distribution in (\ref{eq: sec capacity MRT}) for $M=K=30$ and $t=2,4,8$ antennas, left to right, respectively. The solid line represents the sum of the  first $100$ terms in Theorem~\ref{theorem: mu-me exact outage}.  The dashed and dotted lines represent the distribution upper and lower bounds given in Lemma~\ref{lemma: some eve above thr} and Lemma~\ref{lemma: some bob above thr}, respectively. }
	\label{fig: ratio_approx_distribution}
\end{figure}
For comparison between the  bounds given in Corollary~\ref{coro: meaningful upper bound}, Lemma~\ref{lemma: some eve above thr}, Lemma~\ref{lemma: some bob above thr} and Corollary~\ref{coro: meaningful lower bound}, for $t=4$, $M=K=1000$, see Figure~\ref{fig: coros vs claims}.


\begin{figure}
	\centering
		\includegraphics[scale=0.35]{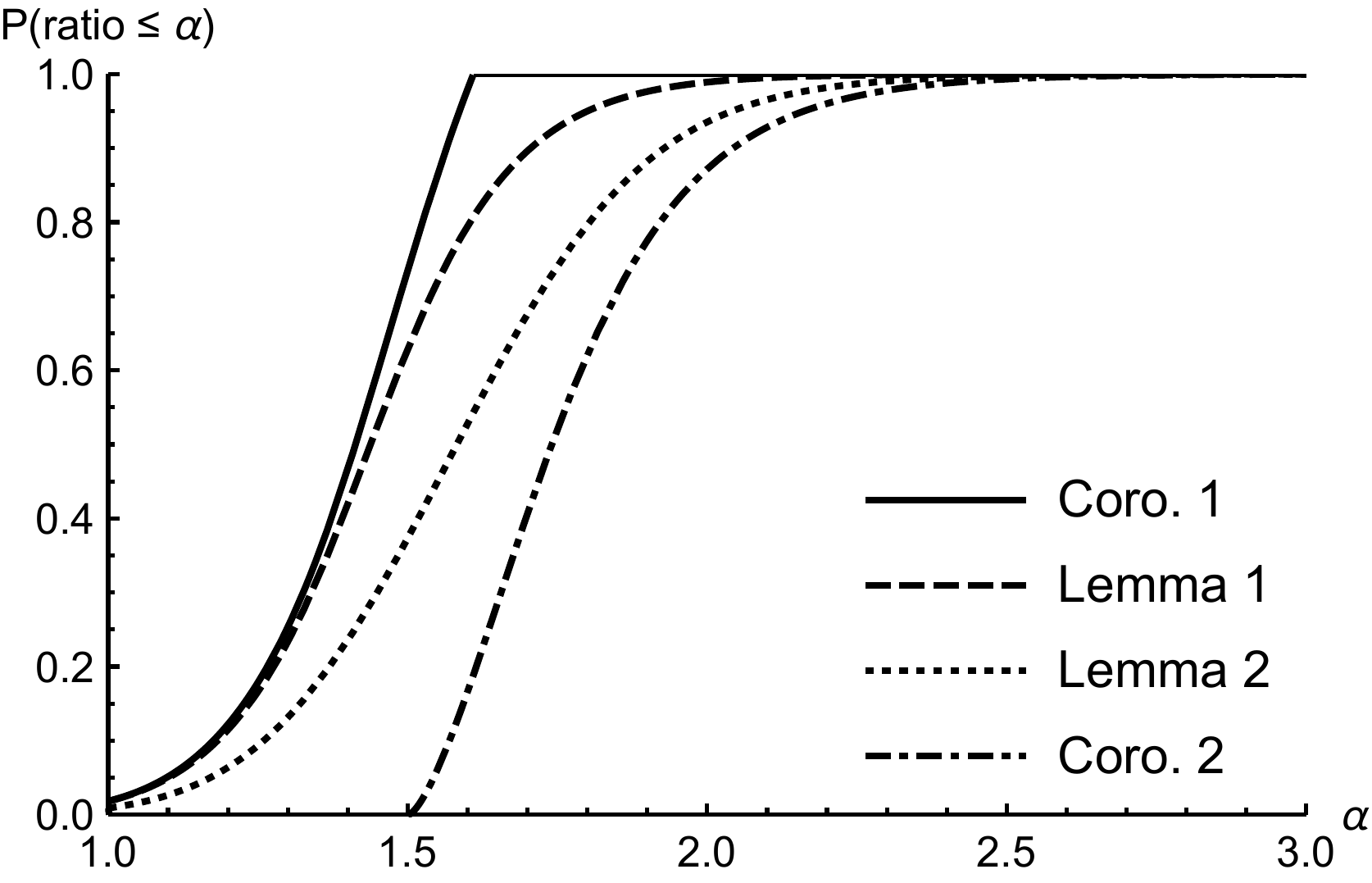}
	\caption{A comparison between the  bounds given in Corollary~\ref{coro: meaningful upper bound}, Lemma~\ref{lemma: some eve above thr}, Lemma~\ref{lemma: some bob above thr} and Corollary~\ref{coro: meaningful lower bound}, respectively, for $t=4$, $M=K=1000$.  }
	\label{fig: coros vs claims}
\end{figure}
Figure~\ref{fig: bob eve ratio} depicts the secrecy outage  probability. In particular, we set  $t=4$ and $\alpha=2$, and fix the number of users to $K=1000$. Then  we examine what is the secrecy outage probability as a function of the number of eavesdroppers. The dots represents the critical ratio between $M$ and $K$ such that $\Lambda (\alpha)=1$, which is exactly $M=\Theta \prnt{K \prnt{\log K}^{t-1}}^{1/\alpha}$. Indeed, for values of $M$ which have smaller order than the critical value, result in small values of $\Lambda(\alpha)$, hence, small outage probability.
\begin{figure}
	\centering
		\includegraphics[scale=0.35]{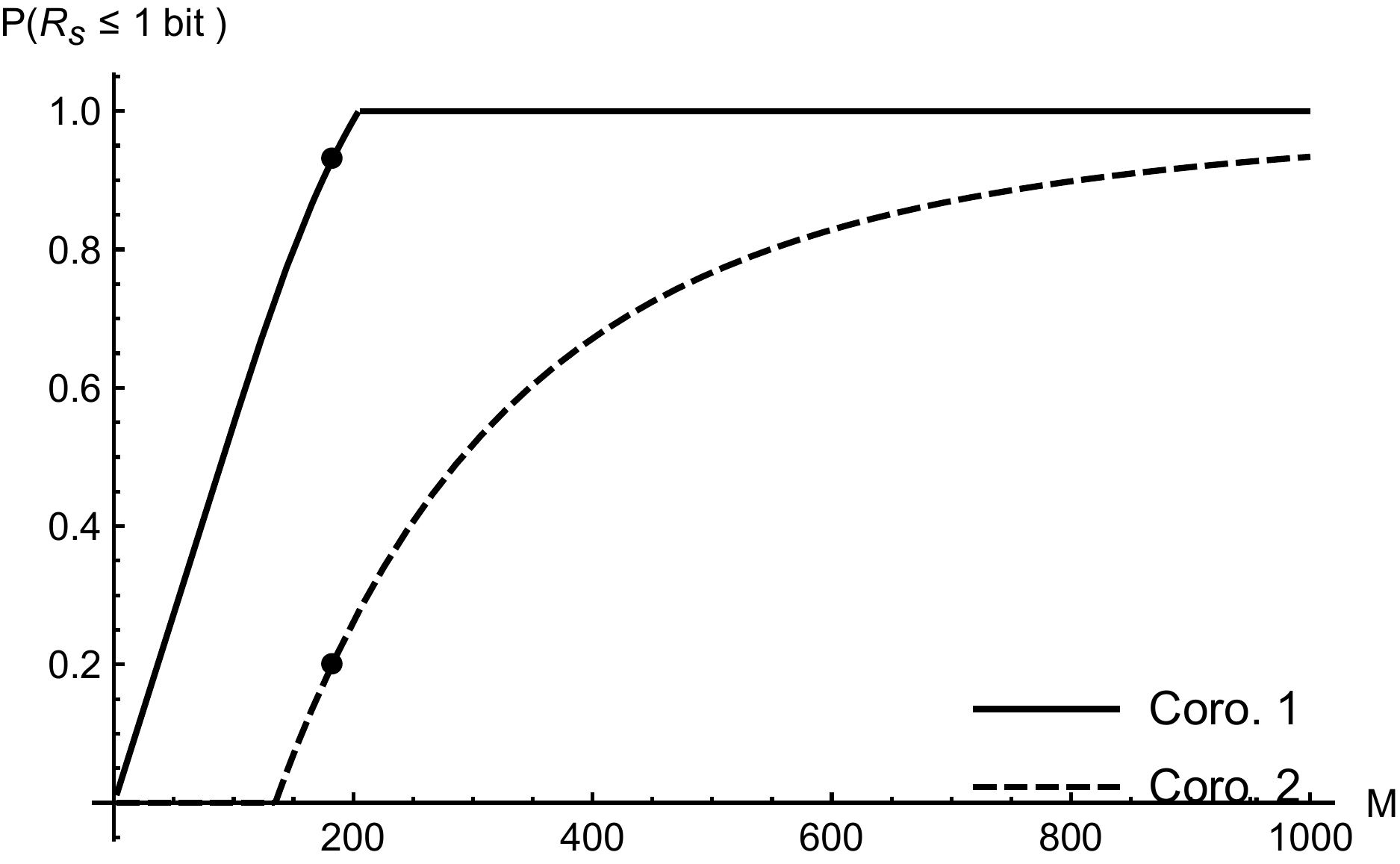}
	\caption{The upper and lower bounds given in Corollary~\ref{coro: meaningful upper bound} and Corollary~\ref{coro: meaningful lower bound}, for $K=1000$, $t=4$ and $\alpha=2$, as a function of the number of eavesdroppers $M$. The marked dot represents the critical ratio where $\Lambda(\alpha)=1$. }
	\label{fig: bob eve ratio}
\end{figure}

\section{Conclusion}
We have studied the secrecy rate and outage probability in the presence of multiple legitimate users and eavesdroppers for the complex Gaussian MISO channel, in which only statistical knowledge on the wiretap channels is available to the transmitter. Specifically, we analyzed the secrecy rate distribution when transmitting to the strongest user while many eavesdroppers are wiretapping, and derived the resulting secrecy outage probability.  We showed that the secrecy rate in such transmission scheme behaves like the ratio of Gumbel distributions, which does not have a closed form expression. Thus, tight upper and lower bounds on the limiting secrecy rate distribution were given. These bounds are tractable and provide insight on the scaling law and the effect of the system parameters on the secrecy capacity.
In particular, the reduction in the secrecy rate as the \emph{number of eavesdroppers grows}, compared to the boost in the secrecy rate as \emph{the number of legitimate users grows} was quantified, and we proved that in the presence of $n$ eavesdroppers, to attain asymptotically small secrecy outage probability with $t$ transmit antennas, $\Omega\prnt{n \prnt{\log n}^{t-1}}$ legitimate users are required.
To support our claims, we conducted rigorous simulations that shows that our bounds are tight.

\section*{Appendix}\label{sec. appendix}
\setcounter{subsection}{0}
\subsection{Proof of Theorem~{\ref{theorem: mu-me exact outage}}}\label{sec. theorem1 proof}
First, note that in (\ref{eq: sec capacity MRT}), the squared norm in the numerator and  the squared inner product in the denominator are independent. Specifically,  while the former represents the length of the user's channel, the latter represents the square of the product of the eavesdropper channel's magnitude and the cosine of the phase between the eavesdropper's channel and the user's channel. Since the angle between i.i.d. Gaussian vectors is independent of their norms \cite{kampeas2013itw}, the distribution of the squared inner product in the denominator is identical for all eavesdroppers and independent of the user index. 

Accordingly, to ease notation, let $\bgamma_{i^*} = P \Vert \bh_{i^*}\Vert^2$  and let $\bgamma_{j^*} = P \vert \langle \hat{\bh}_{i^*}, \bg_{j^*}\rangle \vert^2$. Note that, asymptotically, as both are maximum of series of random variables, they have extreme type distributions, e.g., $\bgamma_{i^*} \sim \mathrm{G}(P a_K, P b_K)$ and $\bgamma_{j^*} \sim \mathrm{G}(P a_M,P b_M)$, where $\mathrm{G}(a_n,b_n)$ denotes the Gumbel distribution with the normalizing constants $a_n$ and $b_n$, given in (\ref{eq: a_n}) and (\ref{eq: b_n}), respectively. 
Further, let us define the ratio transform  $\balpha = \frac{1 + \bgamma_{i^*}}{1 + \bgamma_{j^*}}$ and $\bbeta = \bgamma_{j^*}$,  with the inverse transform, $\bgamma_{i^*}= \balpha(1+\bbeta) -1$. Accordingly, we have, 
\small
\begin{align*}
&\Pr\prnt{\balpha > \alpha , \bbeta > \beta} = \Pr\prnt{ \bgamma_{i^*} > \alpha(1 + \bgamma_{j^*})-1, \bgamma_{j^*}> \beta }&\\
& = \int_\beta^\infty f_{\bgamma_{j^*}}(\gamma_j)\prnt{1 - F_{\bgamma_{i^*}}\prnt{\alpha(1 + \gamma_j)-1}}\mathrm{d}\gamma_j &\\
& = \prnt{1-F_{\bgamma_{j^*}}(\beta)} - \int_\beta^\infty f_{\bgamma_{j^*}}(\gamma_j)F_{\bgamma_{i^*}}\prnt{\alpha(1 + \gamma_j)-1} \mathrm{d}\gamma_j 	 &\\
\end{align*}
\normalsize
Noting that $f_{\bgamma_{j^*}}(\gamma) = \frac{1}{a_M}e^{-\frac{\gamma - b_M}{a_M}}e^{-{e^{-\frac{\gamma - b_M}{a_M}}}}$ and $F_{\bgamma_{i^*}}(\gamma) = e^{-{e^{-\frac{\gamma - b_K}{a_K}}}}$, we exchange variables such that $e^{-\frac{\gamma - b_M}{a_M}} = \zeta$, hence, $\gamma = -a_M \log (\zeta)+ b_M$ and $\mathrm{d}\gamma = \frac{-a_M}{\zeta}\mathrm{d}\zeta$. Further, we note that $\prnt{\zeta \cdot e^{-\frac{b_M}{a_M}}}^{\frac{\alpha a_M}{a_K}}$ is equal to $e^{- \frac{\alpha \gamma}{a_K}}$, which is useful for this case. Accordingly, we have,  
\begin{align*}
&\Pr\prnt{\balpha > \alpha , \bbeta > \beta} = \prnt{1-F_{\bgamma_{j^*}}(\beta)}&\\
& \quad -\int_0^{e^{-\frac{\beta-b_M}{a_M}}} e^{-\zeta} \exp\brcs{-e^{\frac{1+	b_K - \alpha (1+b_M)}{a_K}} \cdot \zeta^{\frac{\alpha a_M}{a_K}}} \mathrm{d}\zeta&
\end{align*}
Generally, this integral cannot be reduced to a closed-form. However, if we replace $e^{-\zeta}$ with its series expansion and noting that we can interchange the sum and the integral from Fubini's theorem, we have, 
\begin{align*}
&\Pr\prnt{\balpha > \alpha , \bbeta > \beta}  = \prnt{1-F_{\bgamma_{j^*}}(\beta)} - \sum_{k=0}^\infty \frac{1}{k!} &\\
& \quad \cdot \int_0^{e^{-\frac{\beta-b_M}{a_M}}} (-\zeta)^k e^{-e^{\frac{1+	b_K - \alpha (1+b_M)}{a_K}} \cdot \zeta^{\frac{\alpha a_M}{a_K}}} \mathrm{d}\zeta&\\
&= \prnt{1-F_{\bgamma_{j^*}}(\beta)}&\\
&\qquad - \frac{a_K }{\alpha a_M}\sum_{k=0}^\infty \frac{(-1)^k}{k!} \prnt{e^{\frac{1+	b_K - \alpha (1+b_M)}{a_K}}}^{-\frac{(k+1)a_K}{\alpha a_M}} &\\
&\qquad \cdot  \Gamma\brkt{\frac{(k+1)a_K}{\alpha a_M},0, e^{-\frac{\alpha(\beta-b_M)}{a_K}} e^{\frac{1+	b_K - \alpha (1+b_M)}{a_K}}}&
\end{align*}
where $\Gamma\brkt{s,0,z} = \int_0^z \tau^s e^{-\tau} \mathrm{d}\tau$ is the lower incomplete Gamma function.
Finally, since we are only interested in the marginal distribution $\balpha$, we set $\beta \to -\infty$ to obtain,
\small
\begin{align}
&\Pr(\balpha \leq \alpha ) \nonumber &\\
&=\frac{a_K }{\alpha a_M}\sum_{k=0}^\infty \frac{(-1)^k}{k!} e^{-(k+1)\frac{1+	b_K - \alpha (1+b_M)}{\alpha a_M}}\cdot\Gamma\left[\frac{(k+1)a_K}{\alpha a_M}\right] \label{eq: infinite sum ex1}&\\
& = \sum_{k=0}^\infty \frac{(-1)^k}{(k+1)!} e^{-(k+1)\frac{1+	b_K - \alpha (1+b_M)}{\alpha a_M}}\cdot \Gamma\left[1+\frac{(k+1)a_K}{\alpha a_M}\right]\label{eq: infinite sum ex2} &
\end{align}
\normalsize
where the last step follows from the identity $x \Gamma[x] = \Gamma[x+1]$ for all $x>0$. Thus, Theorem~\ref{theorem: mu-me exact outage} follows.
\subsection{Proof of Lemma~\ref{lemma: some eve above thr}}
Denote $\mathcal{E} = \brcs{j: \vert \langle \hat{\bh}_{i^*}, \bg_{j}\rangle \vert^2 > u_m}$. Notice that $0 \leq \vert \mathcal{E}\vert \leq M$.

When an eavesdropper sees above-threshold channel projection, then the excess above the threshold follows the exponential distribution with rate $1/a_M$ \cite[Ch. 4.2]{coles2001introduction}. 
To ease notation, let $\bgamma_{i^*} = \Vert \bh_{i^*} \Vert^2$ and let $\bgamma_{\bar{j}} \in \mathcal{E}$. Accordingly, $\bgamma_{i^*} \sim G(a_K, b_K)$ and $\bgamma_{\bar{j}} \sim Exp\brkt{1/a_M}$. Let us define variables transformation $\balpha = \frac{1 + \bgamma_{i^*}}{1 + u_m + \bgamma_{\bar{j}}}$ and $\bbeta = \bgamma_{\bar{j}}$, for which the inverse is $\bgamma_{i^*} = \balpha(1 + u_m + \bgamma_{\bar{j}}) - 1$.  Accordingly, we have
\small
\begin{align*}
&\Pr \prnt{\balpha > \alpha, \bbeta > \beta}&\\
& =  \Pr \prnt{\bgamma_{\bar{j}} > \beta , \bgamma_{i^*} > \alpha\prnt{1 + u_m + \bgamma_{\bar{j}}}-1}&\\
& = \int_\beta^\infty f_{\bgamma_{\bar{j}}}(\gamma_j)\prnt{1 - F_{\bgamma_{i^*}}\prnt{\alpha\prnt{1 + u_m + \bgamma_{\bar{j}}}-1}}\mathrm{d}\gamma_j &\\
& = \prnt{1-F_{\bgamma_{\bar{j}}}(\beta)} &\\
& \qquad - \int_\beta^\infty f_{\bgamma_{\bar{j}}}(\gamma_j)F_{\bgamma_{i^*}}\prnt{\alpha\prnt{1 + u_m + \gamma_{{j}}}-1} \mathrm{d}\gamma_j 	 &
\end{align*} 
\normalsize
Noting that $f_{\bgamma_{\bar{j}}}(\gamma) = e^{-\gamma/a_M}/a_M$ and $F_{\bgamma_{i^*}}(\gamma) = e^{-{e^{-\frac{\gamma - b_K}{a_K}}}}$, we exchange variables such that $e^{-\frac{\alpha \gamma_j}{a_K}} = \zeta$, hence, $\gamma_j = -\frac{a_K}{\alpha} \log (\zeta)$ and $\mathrm{d}\gamma_j = -\frac{a_K}{\alpha \zeta}\mathrm{d}\zeta$. Further, note that $e^{-\frac{\gamma_j}{a_M}}$ is equal to $\zeta^{\frac{a_K}{ \alpha a_M}}$. Thus, we have, 
\begin{align*}
&\Pr\prnt{\balpha > \alpha , \bbeta > \beta} = \prnt{1-F_{\bgamma_{\bar{j}}}(\beta)} &\\
& \qquad -\int_0^{e^{-\frac{\alpha \beta}{a_K}}} \frac{a_K}{\alpha a_M}\zeta^{\frac{a_K}{\alpha a_M} - 1 }\cdot e^{-e^{\frac{1+	b_K - \alpha(1+u_m)}{a_K}} \cdot \zeta} \mathrm{d}\zeta&\\
&=  \prnt{1-F_{\bgamma_{\bar{j}}}(\beta)}  - \frac{a_K}{\alpha a_M} &\\
& \qquad \cdot e^{-\frac{1+ b_K- \alpha(1+u_m)}{\alpha a_M}}\Gamma \brkt{\frac{a_K}{\alpha a_M},0, e^{\frac{1+b_K - \alpha(1+ u_m + \beta)}{a_K}}} 
\end{align*}
To obtain the marginal distribution of the ratio $\balpha$, we set $\beta = 0$, to obtain
\begin{align*}
&\Pr \prnt{\balpha > \alpha} &\\
&= 1 - \frac{a_K}{\alpha a_M} e^{-\frac{1+ b_K - \alpha(1+ u_m )}{\alpha a_M}} \Gamma \brkt{\frac{a_K}{\alpha a_M},0, e^{\frac{1+b_K -\alpha(1+ u_m) }{a_K}}} &
\end{align*}
Note that a $u_m$ can be calculated from the inverse exponential distribution. Namely, $u_m = 2\log(M)$. Note also that for such a threshold the probability that exactly one (strongest) eavesdropper is above-threshold is
\[
\Pr \prnt{\vert \mathcal{E} \vert =1} = \prnt{1-1/M}^{M-1} \to e^{-1} \approx 0.37.
\]
Further, 
\begin{align*}
&\Pr \prnt{\vert \mathcal{E} \vert > 1} = 1 - \prnt{\prnt{1-1/M}^{M-1} + \prnt{1- 1/M}^M }&\\
& \qquad \qquad \quad \to 1 - 2e^{-1} \approx 0.26. &
\end{align*}
Thus, given that some eavesdroppers are above threshold, it is most likely that only one has exceeded. 
\subsection{Proof of Lemma~\ref{lemma: some bob above thr}}
Denote $\mathcal{U} = \brcs{i: \Vert \bh_{i} \Vert^2 > u_k}$. 

When a user sees above-threshold squared channel norm, then the excess above the threshold, in this case as well,  follows the exponential distribution with rate $1/a_K$ \cite[Ch. 4.2]{coles2001introduction}. 
To ease notation, let $\bgamma_{\bar{i}} \in \mathcal{U}$ and let $\bgamma_{j^*} = \vert \langle \hat{\bh}_{i^*}, \bg_{j^*}\rangle \vert^2$. Accordingly, $\bgamma_{j^*} \sim G(a_M, b_M)$ and $\bgamma_{\bar{i}} \sim Exp\brkt{1/a_K}$. Let us define variables transformation $\balpha = \frac{1 + u_k + \bgamma_{\bar{i}}}{1 + \bgamma_{j^*}}$ and $\bbeta = \bgamma_{j^*}$, for which the inverse is $\bgamma_{\bar{i}} = \balpha(1 + \bgamma_{j^*}) - \prnt{1+u_k}$.  Accordingly, we have
\small
\begin{align*}
&\Pr \prnt{\balpha > \alpha, \bbeta > \beta}&\\
& =  \Pr \prnt{\bgamma_{\bar{i}} > 0, \prnt{1 + u_k +\bgamma_{\bar{i}}}/\alpha -1 > \bgamma_{j^*} > \beta}&\\
& = \int_0^\infty f_{\bgamma_{\bar{i}}}(\gamma_i)\prnt{F_{\bgamma_{j^*}}\prnt{\frac{{1 + u_k + \bgamma_{\bar{i}}}}{\alpha}-1} - F_{\bgamma_{j^*}}\prnt{\beta}}\mathrm{d}\gamma_i &\\
& = \int_0^\infty f_{\bgamma_{\bar{i}}}(\gamma_i)F_{\bgamma_{j^*}}\prnt{\frac{{1 + u_k + \bgamma_{\bar{i}}}}{\alpha}-1} \mathrm{d}\gamma_i  - F_{\bgamma_{j^*}}\prnt{\beta}	 &
\end{align*} 
\normalsize
Noting that $f_{\bgamma_{\bar{i}}}(\gamma) = e^{-\gamma/a_K}/a_K$ and $F_{\bgamma_{j^*}}(\gamma) = e^{-{e^{-\frac{\gamma - b_M}{a_M}}}}$, we exchange variables such that $e^{-\frac{\gamma_i}{\alpha a_M}} = \zeta$, hence, $\gamma_i = -\alpha a_M\log (\zeta)$ and $\mathrm{d}\gamma_i = -\frac{\alpha a_M}{ \zeta}\mathrm{d}\zeta$. Further, note that $e^{-\frac{\gamma_i}{a_K}}$ is equal to $\zeta^{\frac{ \alpha a_M}{a_K}}$. Thus, we have, 
\begin{align*}
&\Pr\prnt{\balpha > \alpha , \bbeta > \beta}&\\
& = \int_0^{1} \frac{\alpha a_M}{a_K}\zeta^{\frac{a_K}{\alpha a_M} - 1 }\cdot e^{-\zeta e^{-\frac{1+	u_k - \alpha(1+b_M)}{\alpha a_M}}} \mathrm{d}\zeta - F_{\bgamma_{j^*}}(\beta)&\\
&=  \frac{\alpha a_M}{a_K} e^{\frac{1+ u_k - \alpha(1+b_M)}{a_K}}\Gamma \brkt{\frac{\alpha a_M}{a_K},0, e^{-\frac{1+u_k - \alpha(1+ b_M)}{\alpha a_M}}}&\\
& \qquad - F_{\bgamma_{j^*}}(\beta) & 
\end{align*}
Setting $\beta \to  - \infty$, the marginal distribution of $\balpha$ follows.

Note that a $u_k$ can be calculated from the inverse regularized Gamma function. Note that $u_k$ can also be approximated from $b_K$. Herein, it is also most likely that exactly one (strongest) user exceeded 
for such threshold.

\bibliographystyle{IEEEtran}
\bibliography{bibliography}
\end{document}